%% file: ijcai21.tex
\documentclass{article}
\usepackage{ijcai21}

\input{ijcai_macros.tex}

\usepackage{times}
\usepackage{soul}
\usepackage{url}
\usepackage{xcolor}
\usepackage[hidelinks]{hyperref}
\hypersetup{
    colorlinks=true,
    linkcolor=orange,
    urlcolor=orange,
    citecolor=orange,
}
\usepackage[utf8]{inputenc}
\usepackage[small]{caption}
\usepackage{graphicx}
\usepackage{amsmath}
\usepackage{amsthm}
\usepackage{booktabs}
\usepackage{algorithm}
\usepackage{algorithmic}
\title{Emergent Prosociality in Multi-Agent Games Through Gifting}

\author{
Woodrow Z. Wang$^{1}$\footnote{First three authors have contributed equally and listed randomly.}\and
Mark Beliaev$^{2*}$\and
Erdem B\i y\i k$^{1*}$\and
Daniel A. Lazar$^{2}$\and\\
Ramtin Pedarsani$^{2}$\And
Dorsa Sadigh$^{1}$\\
\affiliations
$^1$Stanford University, 
$^2$University of California, Santa Barbara\\
\emails
\{wwang153, ebiyik, dorsa\}@stanford.edu, \{mbeliaev, dlazar, ramtin\}@ucsb.edu
}
\begin{document}

\maketitle

\begin{abstract}
  Coordination is often critical to forming prosocial behaviors -- behaviors that increase the overall sum of rewards received by all agents in a multi-agent game. However, state of the art reinforcement learning algorithms often suffer from converging to socially less desirable equilibria when multiple equilibria exist. Previous works address this challenge with explicit reward shaping, which requires the strong assumption that agents can be forced to be prosocial. We propose using a less restrictive peer-rewarding mechanism, \emph{gifting}, that guides the agents toward more socially desirable equilibria while allowing agents to remain selfish and decentralized. Gifting allows each agent to give some of their reward to other agents. We employ a theoretical framework that captures the benefit of gifting in converging to the prosocial equilibrium by characterizing the equilibria's basins of attraction in a dynamical system. With gifting, we demonstrate increased convergence of high risk, general-sum coordination games to the prosocial equilibrium both via numerical analysis and experiments.
\end{abstract}

\section{Introduction}

\begin{figure}[t]
    \centering
    \includegraphics[width=\columnwidth]{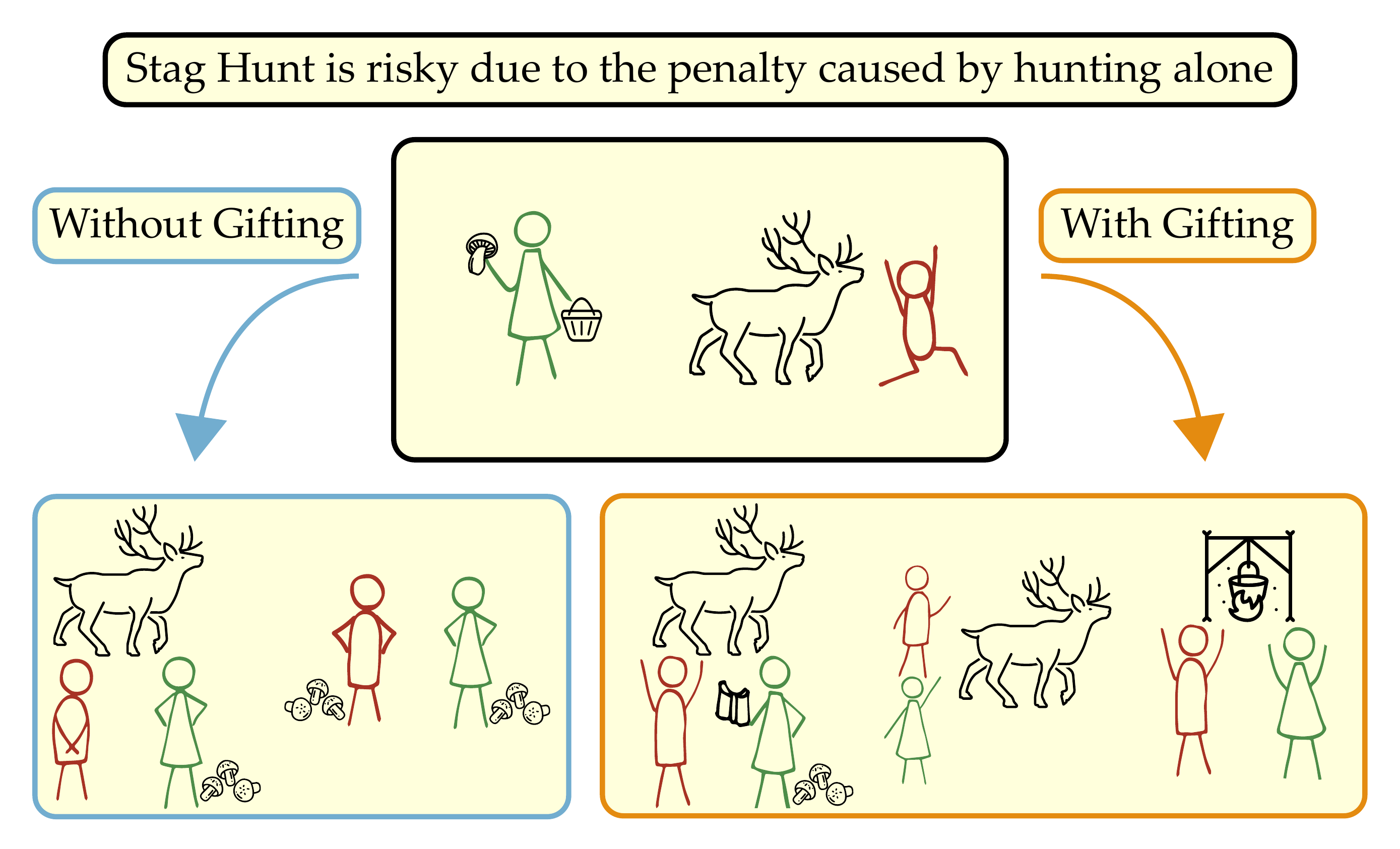}
    \vspace{-22px}
    \caption{Progression of agents in a Stag Hunt game with and without gifting. Agents choose either to hunt or forage. Hunting provides more food, but requires coordination: an agent is severely penalized for hunting alone. Foraging guarantees a small amount of food. Without gifting, agents often learn to forage. With gifting, they gift each other early in training, which mitigates risk, and they learn to hunt together in a prosocial manner.}
    
    %\emph{(Bottom)} Two phase portraits show the dynamics of the game. \emph{(Blue)} Without gifting, agents are more likely to reach the (Forage, Forage) equilibrium due to its larger basin of attraction. \emph{(Orange)} With gifting, the basin of attraction for the (Hunt, Hunt) equilibrium grows, making agents more likely to reach it.}
    \label{fig:front_fig}
    %\Description{Figure showing the progression of learning agents with and without gifting. The basin of attraction of the prosocial equilibrium expands with gifting actions.}
    \vspace{-14px}
\end{figure}

% Removed citation: otherplay
Reinforcement learning (RL) has shown great success in training agents to solve many human-relevant tasks~\cite{mnih2013atari,sutton1999}. In addition, there has been increased interest in leveraging RL techniques in decentralized multi-agent problems, motivated by outstanding performance in two-player zero-sum games such as AlphaZero~\cite{alphazero}. However, simply applying multi-agent RL algorithms to train self-interested agents in a decentralized fashion does not always perform well. Specifically, win-win strategies -- strategies that are beneficial for all agents -- are often challenging to achieve in more advanced settings, such as general-sum games where win-win outcomes are only possible through coordination~\cite{matignon_laurent_le}. 

% Paragraph on prosociality and its connection to coordination
% something like: in these decentralized MA settings, we would like to get to win-win situations -> that is prosociality -> that requires coordination.
% Then after that maybe you can formally say we consider general sum one-shot coordination games and what is a PNE, etc. you shouldn't do all those explanations here. They should go to later sections. Intro should be more motivational and what is the big picture about as opposed let me define all possible types of equilibria and games here.
% Right now it is a bit all over, and the connections between sentences are a bit unclear and the flow of the intro is a bit random
% https://www.overleaf.com/project/5f4fd93a99cecd00017c6e1c
Coordination is often coupled with risk. In the real world, there are many applications where there is a safe action that leads to guaranteed but lower rewards, and a risky action that leads to higher rewards only if agents cooperate, such as alleviating traffic congestion~\cite{biyik2018altruistic,lazar2019learning}, sustainably sharing limited resources~\cite{hughes2018inequity}, and altruistic human-robot interaction~\cite{shirado2018locally}. In a game-theoretic framework, the class of Stag Hunt games is a well-known instance of the tradeoff between social cooperation and safety. In Stag Hunt, two players must independently decide between foraging mushrooms and hunting a stag. If both players choose to hunt the stag, they succeed and are rewarded greatly, while if only one of them goes hunting, they return empty-handed and injured. On the other hand, foraging mushrooms guarantees a meal for the night, although not as satisfying (as shown in Fig. ~\ref{fig:front_fig}). When confronted with this coordination problem, state of the art RL algorithms and humans alike often choose the safer, socially less desirable option of foraging, instead of the riskier, prosocial option of hunting~\cite{vanhyuck,peysakhovich2018prosocial}. This is because uncertainty in the behavior of other players leads one to favor safer alternatives. Ideally, one would prefer reaching the most socially desirable equilibrium -- the \emph{prosocial equilibrium}\footnote{When multiple equilibria are equally prosocial, we refer to reaching any of them as reaching the prosocial equilibrium.} -- to allow all agents to maximize their rewards.

Previous attempts to address the problem of reaching the prosocial equilibrium focus on using explicit reward shaping to force agents to be prosocial, such as by making agents care about the rewards of their partners \cite{peysakhovich2018prosocial}. This requires the strong assumption that a central agent, e.g., a supervisor, can coerce the agents to be altruistic and to care about maximizing the total social utility. We are interested in a less restrictive setting where we do not assume access to a centralized supervisor, and the agents retain their self-interest and only care about maximizing their own received reward -- while the objective is still to increase the probability of reaching the prosocial equilibrium.
% We seek an approach to increase the likelihood of convergence to the prosocial equilibrium in coordination games, while allowing agents to remain selfish. 

\emph{Our key insight is that gifting, a peer-rewarding mechanism, can be used during training as a decentralized teaching mechanism, mitigating risk and encouraging prosocial behavior without explicit reward shaping.} 
Gifting was first introduced by \citet{lupu2020gifting} and is an instance of a larger class of algorithms that extend an agent's action space, providing them with a way to give part of their reward to other agents through their actions. 
In contrast to centralized reward shaping, which requires an external actor to force agents to be prosocial apriori, gifting leaves it up to the agents themselves to use the new actions, enabling self-interested agents to decide \emph{when} and \emph{how} to use the gifting actions.

% In this paper, \emph{our key insight is that gifting, a peer-rewarding mechanism, can be used during training as a decentralized teaching mechanism, mitigating risk and encouraging prosocial behavior without explicit reward shaping.} Gifting was recently introduced by \citet{lupu2020gifting} as a method that allows agents to influence each other's rewards in a multi-agent environment. Gifting extends an agent's action space, providing them with a way to give up a part of their reward to other agents. In contrast to centralized reward shaping, which forces agents to be prosocial apriori, gifting leaves it up to the agent to use the new actions, enabling self-interested agents to decide \emph{when} and \emph{how} to use the gifting actions. 

%While gifting requires that agents can be provided with the ability to share rewards, it does not explicitly require agents to be altruistic; thus, the agents can remain self-interested.

One key advantage of gifting is that it can be used at training time as a behavior shaping mechanism by allowing agents to take risk-mitigating actions, i.e., agent $1$ can gift agent $2$ in order to decrease agent $2$'s risk and incentivize agent $2$ to take a riskier action.  We prove zero-sum gifting does not introduce new pure-strategy Nash equilibria in one-shot normal-form games. 
On the other hand, gifting can introduce new, complex learned behaviors in a repeated normal-form game: new equilibria may be introduced where one agent's policy is contingent on receiving a gift. In this paper, we only demonstrate preliminary results in repeated normal-form games, and focus our analysis on one-shot normal-form games to carefully examine the effects of gifting as a transient risk-mitigating action used only at train time.

% Gifting has two key advantages: (i) At train time, gifting can be used as a behavior shaping mechanism by allowing agents to take risk-mitigating actions, i.e., agent $x$ can gift agent $y$ in order to decrease agent $y$'s risk and incentivize agent $y$ to take a riskier action. (ii) In terms of final learned policies, gifting can also introduce new, complex learned policies, e.g., in a repeated normal-form game, new equilibria may be introduced where one agent's policy is contingent on receiving a gift from another agent. We prove that zero-sum gifting does not introduce new pure-strategy Nash equilibria in one-shot normal-form games. Thus, in order to more carefully examine the effects of gifting as a transient risk-mitigating action used only at train time, we focus the majority of our analysis on one-shot normal-form games, where we can isolate the effects of (i), and motivate (ii) as future work.

% Removed for now, as this seems redundant given the main contributions that follow

%We formalize and investigate the effects of a \emph{zero-sum} gifting mechanism, where an agent adds a fixed amount to others' rewards by losing the same amount from their own reward, on \emph{general-sum} coordination games. We show in this work that  gifting can improve convergence to the prosocial equilibrium without changing the total reward received by all agents in the environment. 

% We show that by adding zero-sum gifting to general-sum one-shot coordination games, selfish agents are more likely to converge the the most prosocial equilibrium. 

Our main contributions in this paper are as follows:
\begin{itemize}[nosep]

\item We propose using a zero-sum gifting mechanism to encourage prosocial behavior in coordination games while allowing agents to remain decentralized and self-interested.

\item We provide insights on the effects of zero-sum gifting for $N$-player one-shot normal-form games by formally showing it does not introduce new equilibria and characterizing conditions under which gifting is beneficial.

% Commented this point out from old paper since this result was moved to the appendix

%\item We provide analytical results demonstrating zero-sum gifting is transient, i.e., the gifting actions occur only during training, and can grow the basin of attraction of the prosocial equilibrium in general-sum one-shot coordination games. 

\item We experimentally show that zero-sum gifting increases the probability of convergence of selfish agents to the prosocial equilibrium in general-sum coordination games when the agents are trained with Deep Q-Networks \cite{mnih2013atari}.

% We show that zero-sum gifting actions are used as transient actions during training and do not introduce new equilibria to single-shot matrix games. Moreover, we empirically show that gift value and risk are positively correlated: as risk increases, the gift value must increase in order to be risk-mitigating. 

\end{itemize}

\section{Related Work}

% Multi-agent environments may have multiple PNE, and it has been an open problem to try to reach the best PNE, or the one with the highest total sum of reward. 

% There has been significant recent work attempting to reach the prosocial equilibrium in coordination games.

%spanning the game theory and RL communities

% \smallskip
% Removed text: \citet{bachrach2009cost} look at the stability of coalitions, agents that cooperate as a group in a game.
\noindent\textbf{Game Theory:}
There has been significant recent work attempting to reach the prosocial equilibrium in coordination games.
Several works use tools from both multi-agent RL and game theory to investigate multi-agent learning in cooperative games. \citet{balcan2015learning} study learning cooperative games from a set of coalition samples. More related to our work, \citet{panait2008theoretical} investigate the effect of leniency, an agent's tolerance of other agents' non-cooperative actions, from an evolutionary game theory perspective. In our work, we investigate non-cooperative games in the scope of gifting. Gifting allows agents to take a new action that may lower the risk their opponent experiences, whereas leniency allows agents to ignore rewards they received in the past. 

%In our work, we examine using a peer-rewarding mechanism that can be easily implemented in multi-agent RL algorithms.

\smallskip

% Removed citations: tversky1992advances,young2012decision

% Removed text: On the other hand, with only groups of two, the payoff-dominant strategy is chosen most of the time. Our goal is to provide RL agents with this ability to reach the payoff-dominant equilibrium -- an equilibrium where at least one agent's payoff is maximized, and no agent would receive a larger payoff if they deviated. 

% \noindent\textbf{Strategic Uncertainty:} There has been work on studying how humans make decisions under conditions of strategic uncertainty~\cite{vanhyuck,kwon2020when}. Specifically, \citet{vanhyuck} show in the setting of repeated coordination games that large groups of humans often fail to coordinate and take a risk-dominated strategy given uncertainty about others' actions. In our work, we also see the effect that risk-dominated strategies have in coordination game settings, and mitigate it by introducing gifting to the environment. 

%initially follow a payoff-dominated strategy -- a strategy that chooses an action with the highest possible payoff -- but because of failure to coordinate, they eventually switch to a risk-dominated strategy -- a strategy that chooses an action with the highest expected payoff, given uncertainty about others' actions. 

\smallskip

% Removed citation: verbeeck2005coordinated,
% Removed text: in order to coordinate each agent's exploration
\noindent\textbf{Coordinated Exploration and Centralized Training:} In multi-agent games, researchers have tried coordinating the exploration of agents to find the most prosocial Nash equilibrium \cite{iqbal2020coordinated} or other equilibrium concepts~\cite{beliaev2020emergent}. While coordinated exploration improves exploration efficiency, it requires communication among agents, which is not available in decentralized settings. Centralized training methods with decentralized control have also been proposed as a way to learn multi-agent policies \cite{lowe2017maddpg,foerster2017coma}. However, similar to coordinated exploration, these approaches require communication among agents during training, which is often not applicable in practice.

%However, it is not realistic in practice to assume agents, who are only self-interested, would accept following a centralized training regimen, given that the game is initially unknown.

% , but it is unclear whether selfish agents would choose to follow a centralized training regimen in practice. 
\smallskip
% Removed citation: huang2019enabling devin2016implemented, zhu2020multi,
\noindent\textbf{Opponent Modeling:} Reasoning about opponent behavior can lead to more complex interactions. With opponent modeling, agents can estimate their opponents' policies, parameters, or updates in order to inform their own learning~\cite{foerster2017lola,sadigh2018planning,sadigh2016planning,shih2021on,xie2020learning,zhu2020multi}. While opponent modeling has shown promising results, it often provides approximation solutions that can be suboptimal. \citet{letcher2018stable} have shown that many opponent modeling methods might prevent convergence to equilibria.

\smallskip
% Removed text: , inspired by~\cite{chentanez2005intrinsically,fehr1999theory}
% One method that has shown success in increasing the probability of convergence to prosocial equilibria is by defining
\noindent\textbf{Explicit Reward Shaping:} To encourage coordination, researchers have explicitly shaped the reward of agents, such as by encoding inequity aversion~\cite{hughes2018inequity}. \citet{peysakhovich2018prosocial} define each agent's reward function to be the sum of all agents' rewards in the environment. Although successful, these approaches require the strong assumption that an agent's reward function can be externally modified and that the agent can be forced to be prosocial and care about maximizing the total utility of all agents. 

\smallskip
% Removed text: It may allow complex behaviors to emerge, such as agents teaching each other which strategies lead to the highest payoffs or rewarding each other for cooperating.
% We show gifting may improve convergence to the prosocial equilibrium in general-sum coordination games. 
\noindent\textbf{Gifting:} Gifting is a recently proposed method that extends the action space of learning agents to allow rewards to be transferred among agents~\cite{lupu2020gifting}. It simply extends each agent's action space with gifting actions, but does not require that the agents use the new gifting actions in any particular way. In our work, we leverage the idea of gifting for improving coordination in general-sum games and examine the effects of the added gifting actions both analytically and experimentally.  

\section{Problem Definition}\label{sec:problem_definition}

We are interested in developing and analyzing algorithms that encourage agents to exhibit prosocial behavior in multi-agent environments with multiple equilibria. We formalize this problem for general-sum one-shot coordination games.

General-sum one-shot coordination games are a class of games with multiple pure-strategy Nash equilibria. \emph{Pure-strategy Nash equilibria} (PNE) are game-theoretic solution concepts, in which each agent has no incentive to unilaterally deviate from a deterministic strategy given the strategy of the other agents. When applying RL techniques to these games, multi-agent systems reach one of the PNE if the agents converge to deterministic policies \cite{Harsanyi}, although not always the best PNE for all agents. Ideally, they would converge to the payoff-dominant PNE, in which at least one player receives a strictly higher payoff and no player would receive a higher payoff in another equilibrium~\cite{Harsanyi}. If such an equilibrium exists, then it is \emph{prosocial} because the sum of rewards for all agents is larger than that of any other equilibrium.
% Removed text: no player would receive a higher payoff in another equilibrium, and at least one player receives a strictly higher payoff in this equilibrium~\cite{Harsanyi}.

% Removed sentence: For example, in Fig.~\ref{fig:basins_of_attraction}~(bottom left), the basin of attraction for the risk-dominant equilibrium is larger due to the risk of taking the prosocial action alone.
% Removed citation: bloembergen2010comparative,
In practice, we are interested in reaching the prosocial equilibrium; however, this is not trivial in settings such as coordination games, where some of the equilibria are risk-dominant, i.e., they have lower but more guaranteed payoffs even if the other agents do not coordinate. These equilibria have larger basins of attraction, so uncertainty in other players' behaviors would lead one to choose the risk-dominant strategy~\cite{Harsanyi}. In settings where both payoff-dominant and risk-dominant equilibria exist, it is difficult to reach the prosocial equilibrium, as agents must be willing to take risks and cooperate with each other. Stag Hunt games are an excellent example of such a setting. Many recent multi-agent RL works have studied Stag Hunt games in depth~\cite{peysakhovich2018prosocial,nica2017learning,leibo2017multi}, as well as works that attempt to build AI systems coordinating with humans~\cite{shum2019theory}. 

% Short Version
% In settings where both payoff-dominant and risk-dominant equilibria exist, it is difficult to reach the prosocial equilibrium, as agents must be willing to take risks and cooperate with each other. Stag Hunt games are an excellent example of such a setting. Many recent multi-agent RL works have studied Stag Hunt games in depth~\cite{peysakhovich2018prosocial,nica2017learning,leibo2017multi}, as well as other works that attempt to build AI systems that coordinate with humans~\cite{shum2019theory}. 

% Long Version
% In settings where both payoff-dominant and risk-dominant equilibria exist, it is difficult to reach the prosocial equilibrium, as agents must be willing to take risks and cooperate with each other. Stag Hunt games, which we will formally introduce in Section~\ref{subsec:stag_hunt}, are an excellent example of such games, incorporating all of these interesting factors. Many recent multi-agent RL works have studied Stag Hunt games in depth~\cite{peysakhovich2018prosocial,nica2017learning,leibo2017multi}, as well as other works that attempt to build AI systems that coordinate with humans~\cite{shum2019theory}. 

%In order to investigate the emergence of prosocial behavior in multi-agent RL through gifting, we focus on general-sum coordination games, which may have a payoff-dominant (prosocial) equilibrium among other potential equilibria. 

The payoff matrix for a general two-action, two-player game is shown in Table~\ref{tab:payoff_matrix}.

%, where lowercase letters denote the payoff of the row agent and the uppercase letters denote the payoff of the column agent.

\begin{table}[ht]
\vspace{-6px}
\begin{center}
\begin{tabular}{ | c | c | c | } 
\hline
 & Action 1 & Action 2 \\ 
\hline
Action 1 & $a,A$ & $b,B$ \\ 
\hline
Action 2 & $c,C$ & $d,D$ \\ 
\hline
\end{tabular}
\vspace{-10px}
\end{center}
\caption{Payoff matrix of a general two-player game}
\vspace{-6px}
\label{tab:payoff_matrix}
\end{table}

In a coordination game, multiple PNE exist, and they occur when players \textit{coordinate} by choosing the same action. This restricts the PNE to lie on the main diagonal of the payoff matrix. Formally, in coordination games, we have $a>c, A>B, d>b, D>C$ ~\cite{Harsanyi}. These inequalities place the PNE on the main diagonal, satisfying the condition that agents must coordinate on the same action in order to reach a PNE. Furthermore, (Action 1, Action 1) is the payoff-dominant equilibrium if $a\geq d, A\geq D$, and at least one of the inequalities is strict. The specific values of these payoffs will determine what sub-class of coordination games is being played: Pure Coordination, Bach or Stravinsky (BoS), Assurance, or Stag Hunt. We are most interested in the Stag Hunt setting because of the difficulty it presents in reaching the prosocial equilibrium. Detailed descriptions of the other sub-class games are in the Appendix. 

\subsection{Stag Hunt}\label{subsec:stag_hunt}
\begin{minipage}{.53\linewidth}
  \centering
    \begin{itemize}
        \item $a>d, A>D$
        \item $a=A,d=D,c=B,$\\$C=b=r$
        \item $a-c < d-r$
    \end{itemize}
\end{minipage}%
\begin{minipage}{.42\linewidth}
  \centering
    \begin{tabular}{ | c | c | c | } 
    \hline
     & Hunt & Forage\\ 
    \hline
    Hunt & $2,2$ & $r,1$ \\ 
    \hline
    Forage & $1,r$ & $1,1$ \\ 
    \hline
\end{tabular}
\end{minipage}
\vspace{5px}

% When referring to coordination games, we are most interested in games that have multiple equilibria which are strictly Pareto-ranked. In the case of two equilibria, they are Pareto-ranked when one of them provides a greater payoff for at least one player without lowering the payoff for any other players. This equilibrium is called the Pareto-dominant equilibrium, and is equivalent to the condition for payoff-dominance. This category of games is defined by \cite{vanhyuck} for $n$ agents, and is analogous to the well known two-player game of Stag Hunt shown above. 

% Removed text: Hence, it provides an interesting setting from both game-theoretic and reinforcement learning perspectives.
Stag Hunt is a two-player game with a risk-dominant equilibrium at (Forage, Forage) and a payoff-dominant equilibrium at (Hunt, Hunt). The payoff-dominant equilibrium is more prosocial and provides each agent with higher reward, but contains risk in the case where the agents do not coordinate. The risk-dominant equilibrium is safer, since the reward is less contingent on the other agent's cooperation. 

% Given the presence of a distinct risk-dominant and payoff-dominant equilibrium, agents are more likely to choose the risk-dominant strategy when uncertain about their opponent's strategy. 

% In practice, we see that agents have a hard time converging to the payoff-dominant equilibrium. 

A Nash equilibrium risk dominates another if it has a strictly higher Nash product ~\cite{Harsanyi}. The Nash product of an equilibrium is the product of deviation losses of both players. As shown in the third condition above $(a-c < d-r)$, all of the parameters influence risk when comparing Nash products of the equilibria. However, for simplicity, we characterize the risk in Stag Hunt with the parameter $r$, the reward for hunting alone, and we keep our analysis focused on $r$ while holding all other values in the payoff matrix constant. In our setting of Stag Hunt, the equilibrium at (Hunt, Hunt) has a Nash product of $(2-1)^2=1$, while the equilibrium at (Forage, Forage) has a Nash product of $(1-r)^2$. With $r$ strictly negative, as $r$ decreases, the risk monotonically increases since the Nash product of (Forage, Forage) grows larger. Thus, we refer to $r$ as the \emph{risk-varying parameter}. In our analysis, we use three versions of Stag Hunt with $r=-2,-6,-10$, referred to as the \emph{low}, \emph{medium}, and \emph{high} risk settings, respectively. If the corresponding setting is not mentioned, then we default to medium risk.

% When agents are uncertain about each other's behavior, they may assume the other agent takes actions uniformly at random. Suppose the column player will choose actions with equal probability. The row player's expected payoff for choosing Hunt and Forage is then $\frac{1}{2}(a+r)$ and $\frac{1}{2}(c+d)$, respectively, so with lower $r$, Hunt becomes suboptimal. The risk of hunting alone would make both learning agents arrive at the risk-dominant equilibrium of (Forage, Forage) as long as $a+r<c+d$.
%In our version of Stag Hunt $a=2$, $c=d=1$, and hence $r$, must be less than zero. 

\subsection{Zero-Sum Gifting}
To increase the probability of reaching the prosocial equilibrium in settings with multiple equilibria, we investigate adding zero-sum gifting actions based on the work by \citet{lupu2020gifting}. With zero-sum gifting, an agent may decide to give some of its reward to the others, preserving the total reward of agents. Although our main interest is with coordination games, the method can be generally applied to all normal-form games and is formalized below. 

%Leveraging this method, we analyze its effects on encouraging agents to coordinate and reach prosocial behavior in Sections~\ref{sec:analysis} and \ref{sec:experiments}.

Take any normal-form game $M$ with a finite set of $N$ players, each with a set of actions (strategies) $S_i$, and payoff function $\mu_i:S_1\times S_2\times\ldots\times S_N\to \R$ where $i\in\{1,2,\dots,N\}$. We denote the subset $\mathbf{S_{PNE}}\subseteq S_1\times S_2\times\ldots\times S_N$ as the set of PNE actions in $M$:
\begin{equation}
    \begin{split}
        &\boldsymbol{s}\in \mathbf{S_{PNE}} \text{ if and only if }\\
        &\forall i\in\{1,2,\ldots,N\} \text{ and  }\forall s_i' \in S_i:
        \mu_i(\boldsymbol{s})\geq \mu_i(s_i',\boldsymbol{s}_{-i}),
    \end{split}
\end{equation}
where $\boldsymbol{s}_{-i}$ denotes the set of actions of all agents other than agent $i$. To introduce gifting, we define a new finite set of actions $G_i$, and function $\sigma_i:G_1\times G_2\times\ldots\times G_N\to \R$ for each player where:
\begin{equation}
    \begin{split}
        &0\in G_i \;,\\
        &\forall g_i\in G_i: g_i\geq 0 \;,\\
        &\sigma_i(\boldsymbol{g}) = -g_i + \frac{1}{N-1}\sum_{j\in -i}g_j \;.
    \end{split}
\end{equation}
Here, $\sigma_i$ formulates how the payoff of agent $i$ changes by the gifting actions of all agents, $\boldsymbol{g}$.

We then formulate the new game $\bar{M}$ with gifting actions. In $\bar{M}$, the set of actions for each player is $\bar{S}_i=S_i\times G_i$, and the corresponding payoffs functions are $\bar{\mu}_i:\bar{S}_1\times \bar{S}_2\times\ldots\times \bar{S}_N\to \R$ where: 
\begin{equation}
    \begin{split}
        &\forall \boldsymbol{\bar{s}}\in \bar{S}_1\times \bar{S}_2\times\ldots\times \bar{S}_N \text{ and } \forall i\in\{1,2,\dots,N\}:\\
        &\bar{\mu}_i(\boldsymbol{\bar{s}})=\mu_i(\boldsymbol{s}) + \sigma_i(\boldsymbol{g}) \text{ where }\boldsymbol{\bar{s}}=(\boldsymbol{s},\boldsymbol{g})\,.
    \end{split}
\end{equation}

Since $\sum_{i=1}^N \sigma_i(\boldsymbol{g}) = 0$, introducing the gifting actions into the game does not change the total reward among all agents.

% Remove Table 2 for now, unless we decide it is useful for readers
%We visualize the payoff matrix of a two-player game with zero-sum gifting where each $G_i = \{0, \gamma\}$ in Table \ref{table:gifting}. The action space of each agent increases from two to four actions. 

Having formalized zero-sum gifting for any normal-form game, we now proceed with analysis and experiments to highlight its benefits. We focus our experiments on settings where we add zero-sum gifting actions with each $G_i = \{0, \gamma\}$.

% \begin{table*}[ht]
% \begin{center}
% \caption{Payoff matrix of general two-player game with zero-sum gifting and $\mathbf{G_i=\{0,\gamma\}}$}
% \vspace{-6px}
% \label{table:gifting}
% \begin{tabular}{ | c | c | c | c | c | }
% \hline
%  & Action 1 & Action 2 & Action 1 + Gift & Action 2 + Gift \\ 
% \hline
% Action 1 & $a,A$ & $b,B$ & $a+\gamma, A-\gamma$ & $b+\gamma,B-\gamma$ \\ 
% \hline
% Action 2 & $c,C$ & $d,D$ & $c+\gamma,C-\gamma$ & $d+\gamma,D-\gamma$ \\ 
% \hline
% Action 1 + Gift & $a-\gamma,A+\gamma$ & $b-\gamma,B+\gamma$ & $a, A$ & $b,B$ \\ 
% \hline
% Action 2 + Gift & $c-\gamma,C+\gamma$ & $d-\gamma,D+\gamma$ & $c,C$ & $d,D$ \\ 
% \hline
% \end{tabular}
% \vspace{-15px}
% \end{center}
% \end{table*}

\section{Analysis of Zero-Sum Gifting in One-Shot Normal-Form Games}
\label{sec:analysis}

We analyze the effect of zero-sum gifting on the equilibria of one-shot normal-form games in Section \ref{subsec:equil}. In Section \ref{subsec:behavior}, we characterize the behavior of learning agents in Stag Hunt with gifting. Specifically, we formulate the learning process of the agents as a dynamical system and show that gifting increases the basin of attraction of the prosocial equilibrium.

\subsection{Effects of Gifting on Equilibria}
% \subsection{Equilibria of One-Shot Normal-Form Games with Gifting}
\label{subsec:equil}
In this section, we state our main theoretical results. We provide the proofs for both Lemma~\ref{lemma1} and Proposition~\ref{proposition1} in the Appendix. We show that agents gift each other $0$ reward in the PNE of $\bar{M}$. Moreover, $\mathbf{\bar{S}_{PNE}}$, the set of PNE in $\bar{M}$, has a one-to-one correspondence with $\mathbf{S_{PNE}}$, the PNE in the original game $M$. Together, these imply having gifting actions does not change the equilibrium behavior of the agents.

\begin{restatable}{lemma}{firstlemma}\label{lemma1}
In any one-shot normal-form game extended with zero-sum gifting actions and for any $s_i\in S_i$, $(s_i,g_i)$ is strictly dominated by $(s_i,0)$ if $g_i\neq0$, meaning $(s_i,0)$ always leads to higher payoff for agent $i$ than $(s_i,g_i)$ for any action profile $\boldsymbol{\bar{s}}_{-i}$ by other agents.
\end{restatable}

\setcounter{theorem}{0} % Erdem: It counts corollaries and lemmas as theorems. We could fix that, but we won't have many theorems, so here we go :D
\begin{corollary}\label{corollary1}
In the set of PNE of any normal-form game extended with zero-sum gifting actions, $\mathbf{\bar{S}_{PNE}}$, all agents gift $0$ reward.
\begin{equation}
\begin{split}
    &\forall\boldsymbol{\bar{s}}\in\mathbf{\bar{S}_{PNE}} \text{ and } \forall i \in\{1,2,\dots,N\}:\\
    &g_i=0\text{ where } \boldsymbol{\bar{s}}=(\boldsymbol{s},\boldsymbol{g})
\end{split}
\end{equation}
\end{corollary}

\vspace{-5px}
\begin{restatable}{proposition}{firstproposition}\label{proposition1}
For any normal-form game $M$ extended to $\bar{M}$ with zero-sum gifting, there exists a unique one-to-one mapping between their corresponding sets of PNE strategy profiles $\mathbf{S_{PNE}}$ and $\mathbf{\bar{S}_{PNE}}$, such that if an action set is a PNE in $M$, then appending $0$-gifting actions gives a PNE in $\bar{M}$:
\begin{equation}
\begin{split}
    &\bigtimes_{i=1}^N s_i \in \mathbf{S_{PNE}} \iff \bigtimes_{i=1}^N (s_i,0) \in \mathbf{\bar{S}_{PNE}}, \text{ and }\\
    &\bigtimes_{i=1}^N (s_i,g_i) \in \mathbf{\bar{S}_{PNE}} \implies \forall i\in\{1,2,\dots,N\}: g_i=0\;.
\end{split}
\end{equation}
\end{restatable}

\vspace{-5px}
% Removed text: We validate this in Section~\ref{subsec:results}.
Proposition~\ref{proposition1} is a desirable result, because it means that introducing gifting actions to one-shot normal-form games will not change the final behavior of the learning agents in the equilibria. Thus, we can carefully investigate gifting's effect on reaching the original equilibria of the game. Moreover, we can view these extended actions as transient to the environment -- the gifting actions are only seen at training time. 

%We will validate this result with simulation experiments in Section~\ref{subsec:results}.
% as we are going to analyze and validate in the subsequent sections,

While not changing the final equilibrium behavior, introducing gifting actions increases the frequency of converging to a more desirable PNE in a dynamic learning environment: prosocial behavior is observed more often after agents converge to an equilibrium. Hence, it is reasonable to extend agents' action spaces with gifting actions specifically in scenarios with higher risk, as we can promote prosocial behavior without directly shaping rewards in a centralized manner.

% \begin{proposition}
% As long as $g > 0$, the pure-strategy Nash Equilibria of the one-stage game does not change, i.e., they are still (Coordinate, Coordinate) and (Defect, Defect).
% \end{proposition}

% \begin{proposition}
% For any gift value $g > 0$, zero-sum gifting actions do not introduce new pure-strategy Nash equilibria to one-shot matrix games. 
% \end{proposition}

% \begin{proof}
% Take an arbitrary one-shot matrix game $G$ defined by the payoff matrix $P$ with $n$ actions $a_1, a_2, ... a_n$ and $m$ agents $A_1, A_2, ... A_m$. With zero-sum gifting actions, we now have a new game $G'$ with $2n$ actions $a_1, a_2, ... a_n, (a_1, Gift), (a_2, Gift), ... (a_n, Gift).$ Suppose for the sake of contradiction that adding zero-sum gifting actions to $G$ introduces a new pure-strategy Nash equilibrium strategy. This new pure-strategy Nash equilibrium must involve some agent using a gifting action, as the payoffs of the original actions remain unchanged. Assume, without loss of generality, that agent $A_1$ should take action $(a_1, Gift)$ in this new equilibrium strategy. Agent $A_1$ has incentive to deviate, as agent $A_1$ can instead take action $a_1$ and receive higher reward. Thus, we have reached a contradiction, as required. 
% \end{proof}

\subsection{Effects of Gifting on the Agents' Behavior}
% \subsection{Behavior of Learning Agents in 2-Player Stag Hunt with Gifting}
\label{subsec:behavior}
Since Stag Hunt is the most interesting and difficult game among the games we introduced in Section~\ref{sec:problem_definition}, we now analyze the behavior of learning agents in Stag Hunt with zero-sum gifting. By Proposition~\ref{proposition1}, we know they will converge to either (Hunt, Hunt) or (Forage, Forage) at PNE\footnote{This game also has a mixed-strategy Nash equilibrium. However, that equilibrium is unstable, so learning agents do not converge there in practice. Thus, we exclude the analysis of this equilibrium.} -- they will not give gifts. In this section, we analyze the deciding factors that lead agents to a specific PNE.

% Removed text: where agents sample their actions from a probability distribution.

Our idea is to formulate the game and the learning process as a dynamical system. To do this, we first define the policies of the two agents, $\pi_{\boldsymbol{x}}$ and $\pi_{\boldsymbol{y}}$, respectively. Here, $\boldsymbol{x}$ and $\boldsymbol{y}$ parameterize the policies. While the policies may have various forms, such as neural networks, what is important is how we define the probability of each action. For that, we let $\boldsymbol{x}$ and $\boldsymbol{y}$ be the logits of a softmax policy:

% which is commonly employed in RL literature \cite{sutton2018reinforcement}:
\begin{equation}
    \begin{split}
        \pi_{\boldsymbol{x}}(s^{(j)}) &= \frac{\exp(x_{j'})}{\sum_{j'=1}^4 \exp(x_{j'})}\;,
    \end{split}
\end{equation}
and similarly for $\pi_{\boldsymbol{y}}$, where $s^{(1)} = \textrm{Hunt}$, $s^{(2)} = \textrm{Forage}$, $s^{(3)} = \textrm{Hunt + Gift}$, $s^{(4)} = \textrm{Forage + Gift}$. As an example, here $\boldsymbol{x}$ and $\boldsymbol{y}$ can be the outputs of a neural network. It can be noted that only the relative differences of the parameters are important: adding a scalar to all parameters of an agent does not change the policy. The two PNEs of the game then correspond to $\forall i\in\{2,3,4\}: x_1-x_i=y_1-y_i=+\infty$ for the prosocial and $\forall i\in\{1,3,4\}: x_2-x_i=y_2-y_i=+\infty$ for the risk-dominant equilibrium.

We then write the expected reward of the agents as:
\begin{equation}
    \mathbb{E}[\bar{\mu}_i] = \sum_{\bar{s}_1\in \bar{S}_1}\sum_{\bar{s}_2\in \bar{S}_2} \pi_{\boldsymbol{x}}(\bar{s}_1)\pi_{\boldsymbol{y}}(\bar{s}_2)\bar{\mu}_i(\bar{s}_1,\bar{s}_2)\;.
\end{equation}

%\begin{figure*}[ht]
%    \centering
%    \includegraphics[width=\textwidth]{figures/portraits.pdf}
%    \vspace{-15px}
%    \caption{Phase portraits of the formulated dynamical system under various gifting parameters. Note the left-most figure shows the system of the game without gifting. Axes show $[-3,3]$ for each plot.}
%    \label{fig:portraits}
    %\Description{Phase portraits of the dynamical system with various gifting parameters are shown.}
%    \vspace{-10px}
%\end{figure*}

While RL methods employ various methods to estimate value functions or policy gradients, true gradients can be closely estimated here by collecting large batches of data at every learning iteration, as this is a one-shot game. Thus, we use the true gradients:
\begin{equation}
    \frac{\partial \mathbb{E}[\bar{\mu}_i]}{\partial x_j} = \sum_{\bar{s}_1\in \bar{S}_1}\sum_{\bar{s}_2\in \bar{S}_2} \frac{\partial \pi_{\boldsymbol{x}}(\bar{s}_1)}{\partial x_j}\pi_{\boldsymbol{y}}(\bar{s}_2)\bar{\mu}_i(\bar{s}_1,\bar{s}_2)\;,
\end{equation}
and similarly for $\frac{\partial \mathbb{E}[\bar{\mu}_i]}{\partial y_j}$. Since both agents are only self-interested and want to maximize their reward, they will update their policies following these gradients.

This formulation leads to an $8$-dimensional autonomous dynamical system (system with no input) with state $\boldsymbol{z}=[\boldsymbol{x}^\top, \boldsymbol{y}^\top]^\top$, and
\begin{equation}
    \dot{\boldsymbol{z}} = f(\boldsymbol{z}) = \left[\frac{\partial\mathbb{E}[\bar{\mu}_1]}{\partial\boldsymbol{x}^\top},\frac{\partial\mathbb{E}[\bar{\mu}_2]}{\partial\boldsymbol{y}^\top}\right]^\top
\end{equation}
We visualize the phase portraits of this dynamical system in Appendix~\ref{sec:gradients}.

While our formulation so far in this section is general and can be applied to any 2-player 4-action (including gifting) normal-form games, we now focus on Stag Hunt. For the remainder of this section, we take $a=A=2$, $r=b=C=-6$, $B=c=1$, $d=D=1$ (the same payoff matrix as in Section~\ref{subsec:stag_hunt} with medium risk) and gift value $\gamma=10$. Similarly, we formulate the original game without gifting as a dynamical system.

%Figure~\ref{fig:portraits} shows the \emph{normalized} gradients of the system that govern the dynamics for various parameters of gifting actions ($x_3$, $x_4$, $y_3$ and $y_4$). Again, as only the differences between parameters are important, we vary the gifting parameters with respect to $x_2$ and $y_2$. Since the prosocial equilibrium is reached with $x_1-x_2=y_1-y_2=+\infty$ and the risk-dominated equilibrium with $x_1-x_2=y_1-y_2=-\infty$, these phase portraits show the two regions of states that would be updated to move towards either of the equilibria. It can be seen that higher gifting parameters enlarge the region that moves towards the prosocial equilibrium.

%While Fig.~\ref{fig:portraits} gives a picture of system dynamics, it is limited in two aspects: first, it does not provide any information about what happens when the gifting parameters are not equal to each other. Second, the gradients of individual states only give information about one-step updates learning agents would have. However, because the gifting parameters will also be learned, Fig.~\ref{fig:portraits} does not show the initial states that will reach the prosocial equilibrium.

We then want to compute the basins of attraction of the equilibria, i.e., the initial states of the system that lead to that specific equilibrium. This is possible, because we already know the stable equilibria of the dynamical systems -- they are the PNEs of the corresponding normal-form games.

To this end, Fig.~\ref{fig:basins_of_attraction} shows the ratio of initial states that reach the prosocial equilibrium when the relative differences of gifting parameters with respect to $x_2$ and $y_2$ are taken as uniformly spaced values in $[-3,3]$\footnote{As only the differences between parameters are important, we vary the gifting parameters with respect to $x_2$ and $y_2$.}. The left heatmap shows the two basins of attraction without gifting. Since the gifting parameters are irrelevant in this setting, the map is binary. The right heatmap corresponds to the extended game with zero-sum gifting actions. The blue line shows the boundary for the game without gifting for comparison. Overall, this shows gifting is indeed beneficial for getting prosocial behavior in the equilibrium.

\begin{figure}[ht]
    \centering
    \vspace{-10px}
    \includegraphics[width=\linewidth]{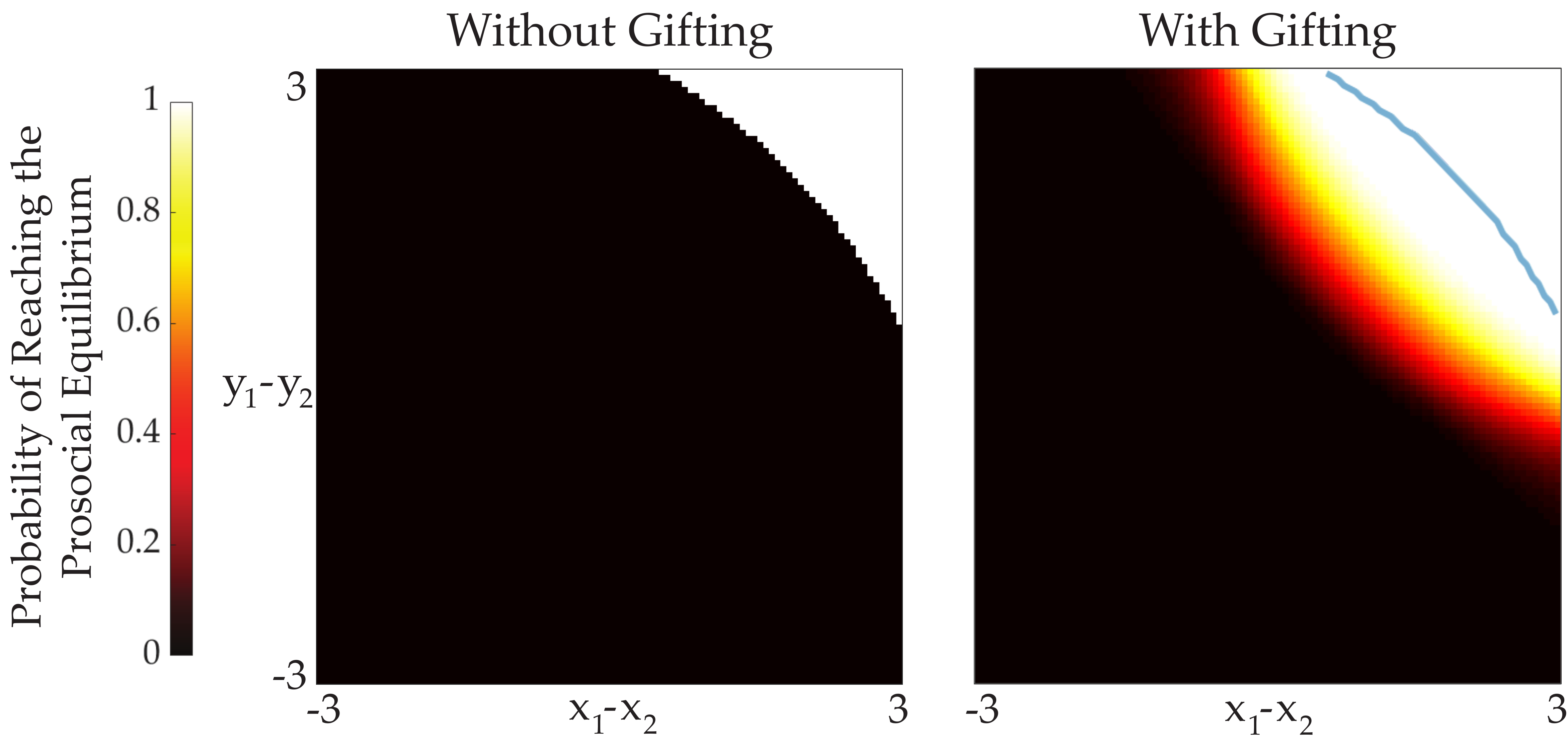}
    \vspace{-15px}
    \caption{Heatmaps show the probability of reaching the prosocial equilibrium when the policy parameters associated with gifting ($x_3-x_2,x_4-x_2,y_3-y_2,y_4-y_2$) are taken uniformly from $[-3,3]$. The blue curve in the right heatmap shows the boundary without gifting for comparison.}
    \vspace{-10px}
    \label{fig:basins_of_attraction}
    %\Description{A $\mathbf{3}$-by-$\mathbf{3}$ grid of plots shows the boundaries that separate the two basins of attraction. Heatmaps of the probability of reaching the prosocial equilibrium with and without gifting are shown.}
\end{figure}

By repeating the same analysis of basin of attraction for varying $\gamma\in\{1,2,\dots,20\}$ and $r\in\{-10,-6,-2\}$, we analyze how often the system converges to the prosocial equilibrium under different conditions. Fig.~\ref{fig:varying_r} suggests that while gifting is always helpful, its benefits become more significant when agents are allowed to gift higher amounts.

\begin{figure}[ht]
    \centering
    \includegraphics[width=\linewidth]{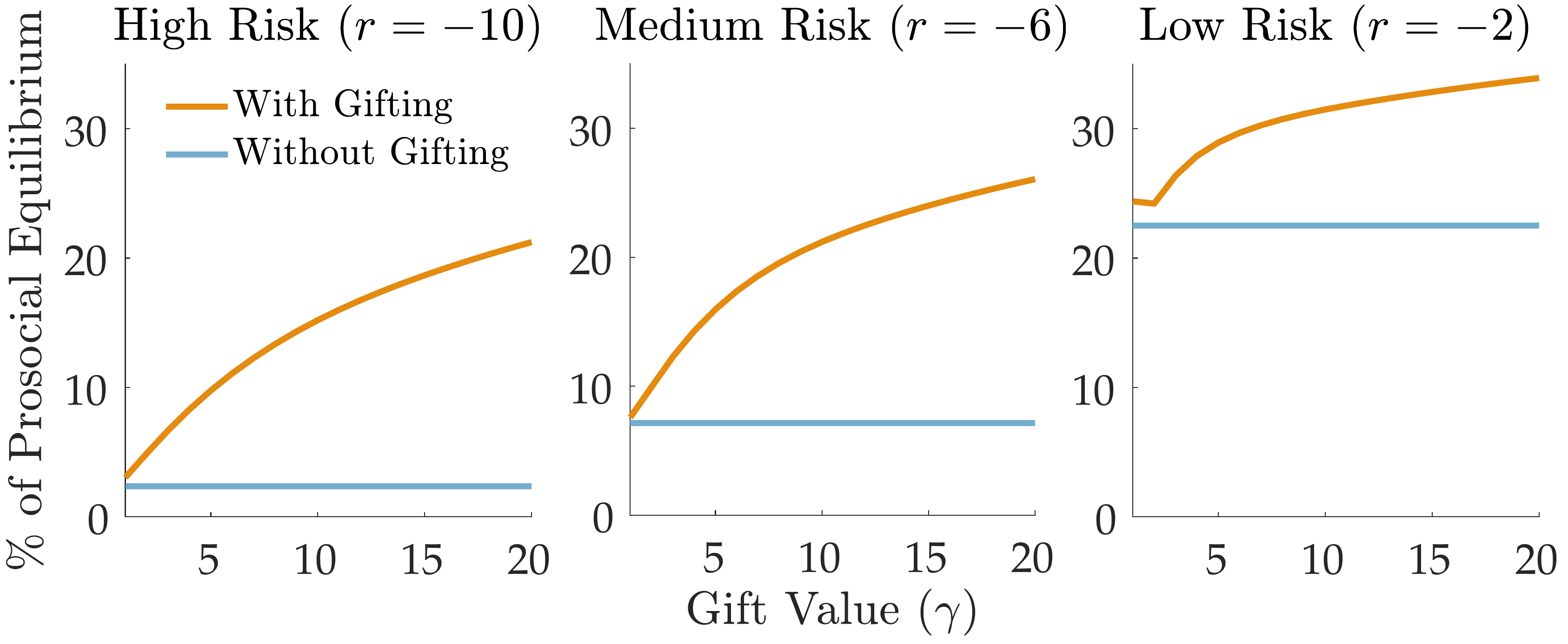}
    \vspace{-18px}
    \caption{Frequency of reaching the prosocial equilibrium under varying risk and gift $\gamma$ amounts. The relative differences of gifting parameters with respect to $x_2$ and $y_2$ are taken as uniformly spaced values in $[-3,3]$ to compute the frequencies.}
    \vspace{-7px}
    \label{fig:varying_r}
\end{figure}

In the next section, we validate these results, which we obtained by assuming access to the true gradients, in more realistic learning settings where the payoffs are unknown to agents.
%~\cite{mnih2013atari}.

\section{Experiments}
\label{sec:experiments}

We first describe the environments used in our experiments and implementation details. We then present experiment results to support our analysis.

% Removed text: and investigate the relation between risk and gift value. 

\subsection{Environments}

\noindent\textbf{Two-Player Environments:} We test the effect of gifting in multiple popular game-theoretic coordination games: Pure Coordination, Bach or Stravinsky, Assurance, and Stag Hunt.

% Removed text: Three versions of Stag Hunt are used, all with varying risk levels derived from the payoff for hunting alone, $r$. Low, medium, and high risk correspond to $r=-2,-6,-10$, respectively.
We perform the majority of our analysis on Stag Hunt (low, medium, and high risk), as it contains both the payoff and risk-dominated equilibria. Assurance is similar to Stag Hunt but the payoff-dominant equilibrium is no longer risk-dominated. BoS and Pure Coordination have two equally prosocial PNE, so they are included in our experiments to demonstrate our method still reaches a PNE in those settings.

\smallskip
\noindent\textbf{$\mathbf{\mathit{N}}$-Player Environments:} We run experiments on Stag Hunt with more than two players, where the game is defined by a graph (similar to \cite{peysakhovich2018prosocial}). Each node in the graph represents an agent and the edges define the individual Stag Hunts to be played. Thus, each agent chooses an action to play with all neighbor agents and receives the average reward of the games. We specifically examine three-player and four-player fully connected graphs in the medium risk setting: FC-3 Stag Hunt and FC-4 Stag Hunt.

\smallskip
\noindent\textbf{Repeated Games:} We further investigate the effect of zero sum gifting in repeated interactions by running experiments on a Repeated Stag Hunt environment. In Repeated Stag Hunt, agents repeatedly play the medium risk, one-shot Stag Hunt over a finite horizon $H\!=\!10$. Each agent observes the most recent action taken by the other agent. This setting is interesting as agents have an expanded policy space: their policies are conditioned on other players' previous actions, and so new Nash equilibria may emerge with gifting. 

% \smallskip
% Table \ref{tab:tables for other games} summarizes the results on these environments.

\subsection{Implementation Details} 

We use the payoff matrices shown in Section~\ref{sec:problem_definition}. Unless otherwise stated, we set $\gamma=10$. For all experiments, we train a Deep Q-Network (DQN) with independent $\epsilon$-greedy exploration for each agent. We use Adam optimizer with a learning rate of $5\times10^{-4}$. The replay buffer size is $10^5$. The $\epsilon$ for exploration begins at $0.3$ and exponentially decays to $0.01$ over $2\times 10^4$ steps. Each target network updates every $250$ episodes. For the one-shot games, all agents are given a constant observation of 0. We provide supplementary code for reproducibility of all the experiments. 

%We leave the rest of the implementation details in the supplementary code provided.

\begin{table}[ht]
\begin{center}

%\begin{tabular}{ | c | c | c | c | c | } 
\begin{tabular}{lcc}\toprule
%\hline
\textit{Environment} & \textit{Without Gifting} & \textit{With Gifting} \\ \midrule
%\hline
Bach or Stravinsky & $\mathbf{100.0\%}$ & $\mathbf{100.0\%}$ \\ 
%\hline
Pure Coordination & $\mathbf{100.0\%}$ & $\mathbf{100.0\%}$ \\ 
%\hline
Assurance & $56.8\%$ & $\mathbf{63.3\%}$ \\ 
%\hline
High Risk Stag Hunt & $0.0\%$ & $\mathbf{19.0\%}$ \\ 
Med. Risk Stag Hunt & $8.6\%$ & $\mathbf{21.4\%}$ \\ 
Low Risk Stag Hunt & $\mathbf{25.4\%}$ & $22.0\%$ \\ 
%\hline
FC-3 Stag Hunt & $7.8\%$ & $\mathbf{12.1\%}$ \\ 
%\hline
FC-4 Stag Hunt & $5.1\%$ & $\mathbf{7.4\%}$ \\
Repeated Stag Hunt & $0.0\%$ & $\mathbf{19.7\%}$ \\\bottomrule
%\hline
\end{tabular}
\vspace{-5px}
\caption{The percentage of $\mathbf{1024}$ runs with random initializations that reached the prosocial equilibrium with multi-agent DQN.}
\vspace{-10px}
\label{tab:tables for other games}
\end{center}
\end{table}
\vspace{-5px}
\subsection{Results}\label{subsec:results}

As shown in Table \ref{tab:tables for other games}, zero-sum gifting increases the probability of converging to the most prosocial equilibrium in a variety of coordination games. In BoS and Pure Coordination, all equilibria are equally prosocial, and we converge to one of the equilibria $100\%$ of the time with and without gifting. In Assurance, the lack of risk makes the prosocial equilibrium a favourable outcome even without gifting, but we still see an improvement when adding gifting actions to the agents. In Stag Hunt we see that gifting has a greater benefit when risk is higher, but performance diminishes slightly in the low risk setting when gifting is introduced. This interdependence between varying risk and gifting is further explored later in this section. In the FC-3 and FC-4 Stag Hunts, we can see that gifting helps increase the probability of convergence to the prosocial equilibrium, but as the number of agents increases, it becomes more difficult to coordinate all agents and encourage risky prosocial behavior over safer actions. In Repeated Stag Hunt, gifting significantly increases the probability of convergence to the prosocial equilibrium. When compared to the results of the corresponding one-shot medium risk Stag Hunt, we can see the likelihood of agents coordinating at the prosocial equilibrium decreases both with and without gifting, implying that coordination over repeated instances of Stag Hunt is a more difficult, risky setting.

\smallskip
\noindent\textbf{Interdependence of Risk and Gift Value:} In Fig. \ref{fig:interdependence_risk_gift}, we examine the relation between the gift value and the risk value in Stag Hunt. The results show that in order for gifting to help in coordination games with higher risk, the gift value needs to increase to compensate for the added risk. 

%  As the payoff for hunting alone decreases, the risk for the game increases.
%As gift value increases, the gifting actions become better at mitigating risk.  

We can also see a slight negative effect gifting has in low risk settings when training with DQN agents. One explanation for this is that adding gifting actions to agents expands their action space and makes exploration more difficult, and uncertainty in the other agent's actions favors the risk-dominant equilibrium. However, under high risk settings, agents are more likely to behave prosocially with gifting, even with increased difficulty in exploration. 

\begin{figure}[t]
    \centering
    %\vspace{-12px}
    \includegraphics[width=0.8\columnwidth]{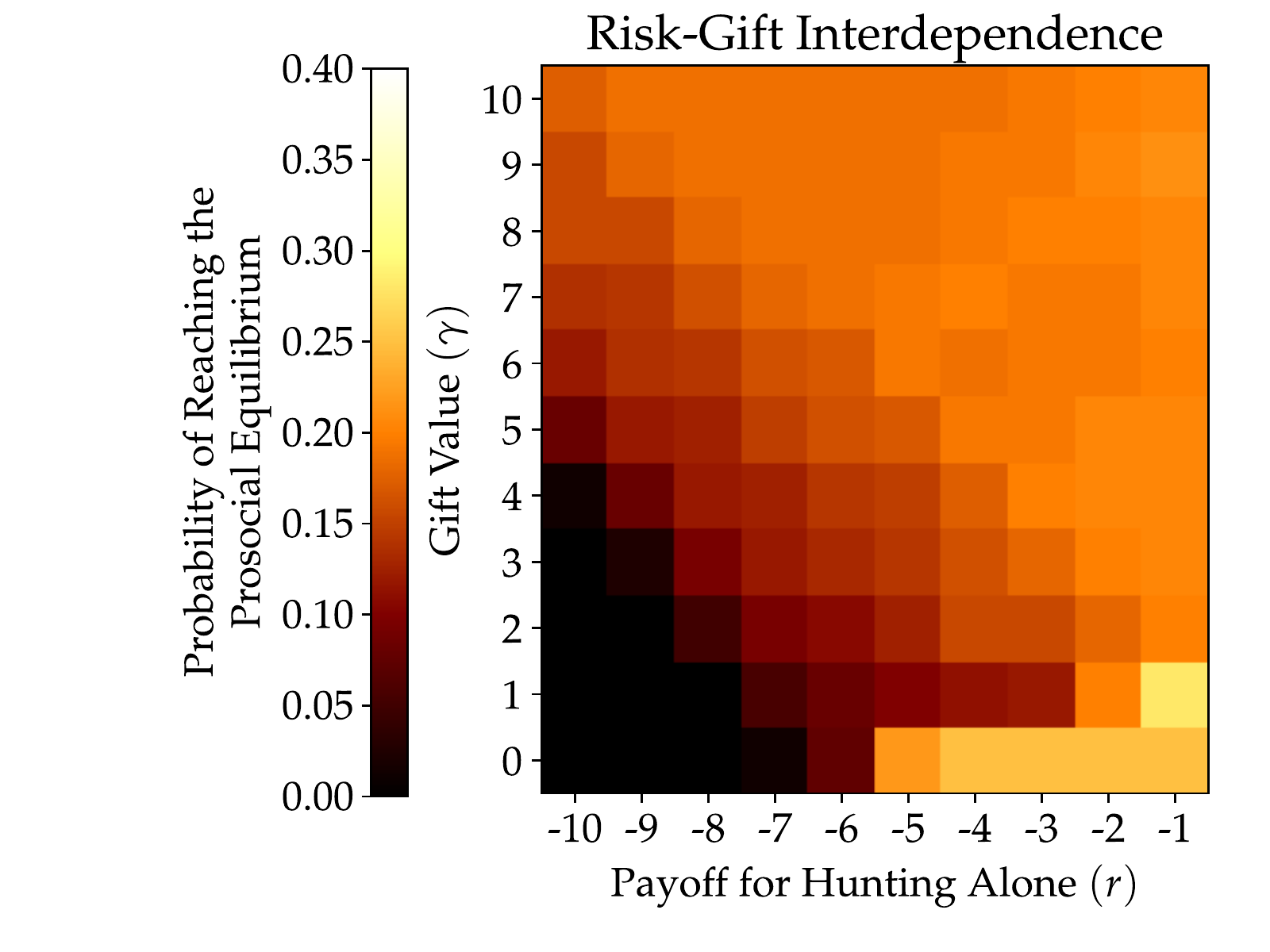}
    \vspace{-15px}
    \caption{This figure depicts the relationship between the risk and gift value in Stag Hunt. In Stag Hunt, the payoff for hunting alone $\mathbf{(r)}$ characterizes the risk. The results show that as risk increases, the gift value $\mathbf{\gamma}$ must increase proportionally in order to be risk-mitigating and improve convergence to the prosocial equilibrium.}
    \label{fig:interdependence_risk_gift}
    %\Description{A heatmap of the relationship between risk and gift value is shown.}
    \vspace{-12px}
\end{figure}

\section{Discussion}

\noindent\textbf{Summary:} We formalize a zero-sum gifting mechanism and show that it often increases the probability of convergence to the prosocial equilibrium in coordination games. We prove that zero-sum gifting does not alter the behavior under Nash equilibria in one-shot normal-form games. With gifting, we show via numerical analysis that the prosocial equilibrium's basin of attraction grows in Stag Hunt and empirically validate these results with DQN in a broader set of environments.

\smallskip
\noindent\textbf{Limitations:} We analyze gifting as an alternative method for encouraging prosocial behavior compared to explicit reward shaping. In practice, gifting requires the ability to extend an environment's action space, so it can only be applied in settings where agents' action spaces can be modified.

% Our experiments demonstrate that gifting can be beneficial in many coordination games. However, our numerical analysis in Section~\ref{subsec:behavior} is limited to a specific case of Stag Hunt. While we observed similar results (and even better results with higher gift amounts $\gamma$) under other Stag Hunt instances, more analyses are required for other coordination games.

Moreover, although our experimental results in Table \ref{tab:tables for other games} show that gifting negatively affects the low risk Stag Hunt setting when trained with DQN, the performance loss is marginal compared to the performance gain we see in higher risk settings. Nonetheless, one should be cautious when applying gifting, as the benefits are dependent on the risk in the respective environment.

\smallskip
\noindent\textbf{Future Work:} We focus the majority of our experiments on one-shot games, since we are interested in isolating the setting where no new equilibria are introduced by the gifting actions. We provide brief experiments of gifting in the repeated game setting, but further exploring the emergence of complex behaviors involving gifting in repeated interactions can help shed light on what settings gifting would be most beneficial. 

%our theoretical claims from Proposition~\ref{proposition1} and Lemma~\ref{lemma1} no longer hold. Further exploring the emergence of these complex behaviors involving gifting can help shed light on what settings gifting would be most beneficial. 

%Finally, it would be interesting to combine gifting with opponent modeling, as the benefits of gifting can be amplified if agents learn to intelligently gift as a teaching mechanism to deliberately influence other agents' rewards. 
\section*{Acknowledgments}
We would like to thank NSF EPCN grant \#1952920 and the DARPA HiCon-Learn project for their support.

\bibliographystyle{named}
\balance\small\bibliography{refs}

\clearpage
\appendix
\section{Sub-classes of Coordination Games}
We formally define each sub-class of coordination game that we provide experimental results for in Table \ref{tab:tables for other games}. We define the conditions for each coordination game, as well as provide a concrete example in the form of a payoff matrix. 
\subsection{Pure Coordination} 

\begin{minipage}{.435\linewidth}
  \centering
    \begin{itemize}
        \item $b=B=c=C=\zeta$
        \item $a=d, A=D$
        \item $\zeta<\min(a,A)$
    \end{itemize}
\end{minipage}%
\begin{minipage}{.52\linewidth}
  \centering
    \begin{tabular}{ | c | c | c | } 
    \hline
     & Action 1 & Action 2\\ 
    \hline
    Action 1 & $1,1$ & $0,0$ \\ 
    \hline
    Action 2 & $0,0$ & $1,1$ \\ 
    \hline
\end{tabular}
\end{minipage}
\vspace{5px}

In the simplest type of coordination game, there does not exist a payoff-dominant equilibrium, as both PNE give identical payoffs and are equally prosocial. There is also no additional risk associated with choosing Action 2 as opposed to Action 1, and hence, we expect randomly initialized learning agents to converge to either equilibrium with equal probability. 

\subsection{Bach or Stravinsky (BoS)}
\begin{minipage}{.435\linewidth}
  \centering
    \begin{itemize}
        \item $b=B=c=C=\zeta$
        \item $a>d, A<D$
        \item $\zeta<\min(a,A)$
    \end{itemize}
\end{minipage}%
\begin{minipage}{.52\linewidth}
  \centering
    \begin{tabular}{ | c | c | c | } 
    \hline
     & Action 1 & Action 2\\ 
    \hline
    Action 1 & $2,1$ & $0,0$ \\ 
    \hline
    Action 2 & $0,0$ & $1,2$ \\ 
    \hline
\end{tabular}
\end{minipage}
\vspace{5px}

In BoS, the PNE are not identical, and neither PNE is payoff-dominant. The row player prefers (Action 1, Action 1) and the column player prefers (Action 2, Action 2). When $a=D$, $A=d$, we consider either PNE to be the prosocial equilibrium, as the sum of rewards among players is identical. In this case, both equilibria have the same risk, and we expect randomly initialized learning agents to converge to either equilibrium with equal probability. 

\subsection{Assurance}
\begin{minipage}{.435\linewidth}
  \centering
    \begin{itemize}
        \item $b=B=c=C=\zeta$
        \item $a>d, A>D$
        \item $\zeta<\min(a,A)$
    \end{itemize}
\end{minipage}%
\begin{minipage}{.52\linewidth}
  \centering
    \begin{tabular}{ | c | c | c | } 
    \hline
     & Action 1 & Action 2\\ 
    \hline
    Action 1 & $2,2$ & $0,0$ \\ 
    \hline
    Action 2 & $0,0$ & $1,1$ \\ 
    \hline
\end{tabular}
\end{minipage}
\vspace{5px}

In the game of Assurance, the payoff-dominant PNE is (Action 1, Action 1). It is also risk-dominant, i.e., even if an agent thinks their partner may not coordinate by taking the same action, there is no incentive to go for Action 2. This makes it easy for agents to reach the payoff-dominant equilibrium.

\section{Proof of Lemma 1}

\firstlemma*

\begin{proof}
For any $\boldsymbol{\bar{s}}_{-i}=(\boldsymbol{s}_{-i},\boldsymbol{g}_{-i})\in\boldsymbol{\bar{S}}_{-i}$, the payoff for agent $i$ under the action $(s_i,g_i)$ is
\begin{equation}
\begin{split}
    \bar{\mu}_i(\boldsymbol{\bar{s}}) &= \mu_i(\boldsymbol{s}) + \sigma_i(\boldsymbol{g}) = \mu_i(s) - g_i + \frac{1}{N-1}\sum_{j\in -i}g_j
\end{split}
\end{equation}
If agent $i$ had $(s_i,0)$, its payoff would be $\mu_i(\boldsymbol{s}) + \frac{1}{N-1}\sum_{j\in -i}g_j$, which is strictly larger as $g_i>0$, regardless of $\boldsymbol{\bar{s}}_{-i}$. Hence, $(s_i,g_i)$ is strictly dominated by $(s_i,0)$, and this completes the proof.
\end{proof}

\section{Proof of Proposition 1}
\firstproposition*
\begin{proof}
We already know from Corollary~\ref{corollary1} that actions involving non-zero gifting cannot exist in the PNE of $\bar{M}$, so the latter statement is true.
We now prove the former statement. First, we define $\mathbf{\bar{S}_0}$ by appending $0$-gifting actions to the action sets in $\mathbf{S_{PNE}}$:
\begin{equation}
    \mathbf{\bar{S}_0} = \left\{\bigtimes_{i=1}^N (s_i,0) \mid \boldsymbol{s}\in\mathbf{S_{PNE}}\right\}\;.
\end{equation}
Next, we show that $\mathbf{\bar{S}_0} = \mathbf{\bar{S}_{PNE}}$, proving the first statement.
\begin{equation*}
    \begin{split}
        &\forall \boldsymbol{\bar{s}}\in \mathbf{S_0} \text{ and } \forall i\in\{1,2,\ldots,N\}:\\
        &\bar{\mu}_i(\boldsymbol{\bar{s}})=\mu_i(\boldsymbol{s})+\sigma_i\left((0,0,\dots,0)\right)=\mu_i(\boldsymbol{s})+0
    \end{split}
\end{equation*}
Because $\boldsymbol{s}\in\mathbf{S_{PNE}}$, we have
\begin{equation*}
    \forall s_i' \in S_i:\mu_i(\boldsymbol{s})\geq \mu_i(s_i',\boldsymbol{s}_{-i})\;,
\end{equation*}
implying changing $s_i$ only does not increase the payoff for agent $i$. Moreover, we know from Lemma~\ref{lemma1} that any non-zero gifting action is strictly dominated by the corresponding zero-gifting action. These two results mean changing the gifting action $g_i$, the original action $s_i$, or both cannot increase the payoff for agent $i$:
\begin{equation}
    \forall \bar{s}_i'\in \bar{S}_i: \bar{\mu}(\boldsymbol{\bar{s}})\geq \bar{\mu}(\bar{s}_i',\boldsymbol{\bar{s}}_{-i})\;,
\end{equation}
and therefore $\mathbf{\bar{S}_0} = \mathbf{\bar{S}_{PNE}}$.
\end{proof}

\section{Gradients of the Dynamical System}\label{sec:gradients}
\begin{figure*}[h]
    \centering
    \includegraphics[width=\textwidth]{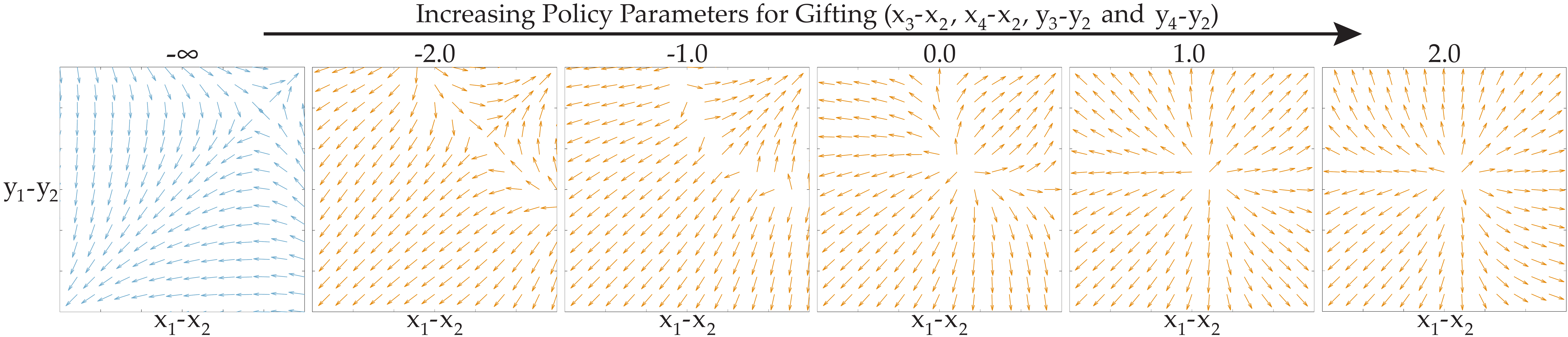}
    % \vspace{-15px}
    \caption{Phase portraits of the formulated dynamical system under various gifting parameters. Note the left-most figure shows the system of the game without gifting. Axes show $[-3,3]$ for each plot.}
    \label{fig:portraits}
    %\Description{Phase portraits of the dynamical system with various gifting parameters are shown.}
    % \vspace{-10px}
\end{figure*}

Figure~\ref{fig:portraits} shows the \emph{normalized} gradients of the system that govern the dynamics for various parameters of gifting actions ($x_3$, $x_4$, $y_3$ and $y_4$). Again, as only the differences between parameters are important, we vary the gifting parameters with respect to $x_2$ and $y_2$. Since the prosocial equilibrium is reached with $x_1-x_2=y_1-y_2=+\infty$ and the risk-dominated equilibrium with $x_1-x_2=y_1-y_2=-\infty$, these phase portraits show the two regions of states that would be updated to move towards either of the equilibria. It can be seen that higher gifting parameters enlarge the region that moves towards the prosocial equilibrium.

It should be noted that while Fig.~\ref{fig:portraits} gives a picture of system dynamics, it is limited in two aspects: first, it does not provide any information about what happens when the gifting parameters are not equal to each other. Second, the gradients of individual states only give information about one-step updates learning agents would have. However, because the gifting parameters will also be learned, Fig.~\ref{fig:portraits} does not show the initial states that will reach the prosocial equilibrium. Therefore, the basin of attraction analyses we made in Section~\ref{subsec:behavior} gives a more accurate picture.

\section{Transient Gifting Actions}

\begin{figure}
    \centering
    % \vspace{-10px}
    \includegraphics[width=\columnwidth]{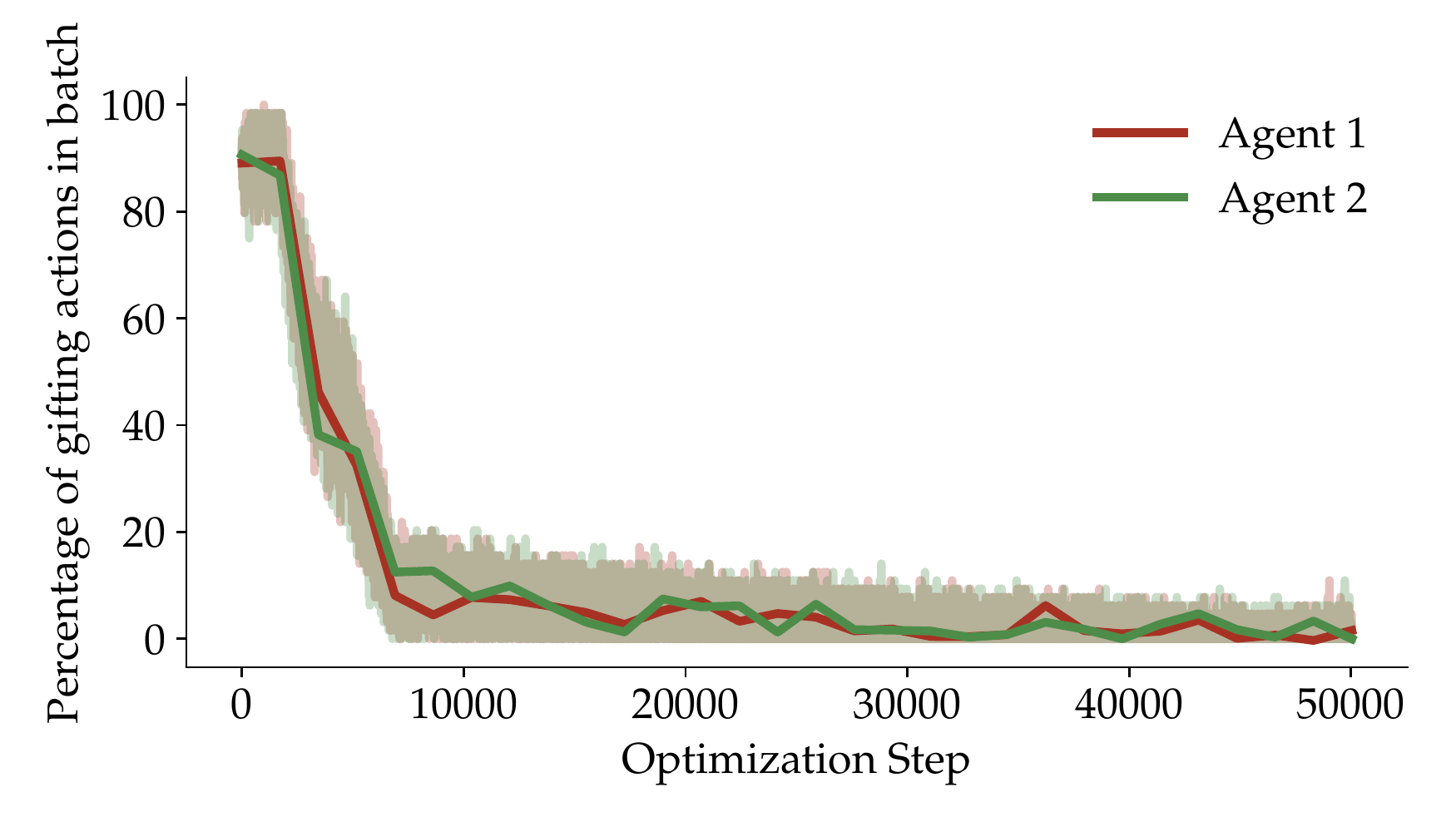}
    % \vspace{-10px}
    \caption{This plot shows the percentage of gifting actions in a batch vs. the optimization step during one training run that starts with frequent gifting actions and reaches the prosocial equilibrium in Stag Hunt. Both agents are initialized to have a higher Q-value for the gifting actions and equal Q-value for the non-gifting actions. This provides an empirical example of Lemma 1 in practice, where agents gift initially to encourage prosocial behavior, and learn not to gift in the limit, while reaching the prosocial equilibrium. }
    \label{fig:analyzegifting}
    %\Description{A plot of the usage of gifting actions throughout training is shown.}
    % \vspace{-10px}
\end{figure}

As Proposition 1 shows, zero-sum gifting does not introduce any new equilibria to one-shot normal-form games. Thus, we investigate the usage of gifting actions at train time to provide insight on how gifting encourages agents to be prosocial. Fig. \ref{fig:analyzegifting} shows that, even when agents start with frequent zero-sum gifting actions, they use gifting as transient actions during training to encourage other agents to update towards the more prosocial equilibrium. As the agents optimize their own parameters selfishly, the agents take the gifting actions less frequently, and their final converged policies never include gifting actions in the case of one-shot normal-form games. 

\section{Compute Details}

The basin of attraction code ran on an Elastic Compute Cloud (EC2) instance in Amazon Web Services (AWS) with 16 vCPUs and 30 GB RAM. Each run took between 2 and 24 hours depending on how fast the agents converge to equilibria.

The DQN training code ran on a personal computer with an 8C/16T processor and 32 GB RAM. Figure \ref{fig:interdependence_risk_gift} took 36 hours to complete. Each result in Table \ref{tab:tables for other games} took around 2 hours to complete.

\end{document}

%% file: ijcai_macros.tex
\usepackage{amsmath,amsfonts,bm}

\usepackage{url}
\usepackage{algorithmic}
\usepackage{graphicx}
\usepackage{amsmath,amssymb,amsthm}
\usepackage{thmtools,thm-restate} % For repeat numbering in Appendix proof
\usepackage{mathtools}

\newtheorem{corollary}{Corollary}
\usepackage{balance} % for balancing columns on the final page

\usepackage{enumitem}
\usepackage{titlesec}
\titlespacing*{\subsection}{0pt}{0.3\baselineskip}{0.1\baselineskip}
\titlespacing*{\section}{0pt}{0.6\baselineskip}{0.5\baselineskip}
\usepackage[nodisplayskipstretch]{setspace}

\newcommand{\citet}[1]{\citeauthor{#1}~\shortcite{#1}}

\def\eqref#1{equation~\ref{#1}}
\def\1{\bm{1}}

\DeclareMathAlphabet{\mathsfit}{\encodingdefault}{\sfdefault}{m}{sl}
\SetMathAlphabet{\mathsfit}{bold}{\encodingdefault}{\sfdefault}{bx}{n}

\newcommand{\R}{\mathbb{R}}